\documentclass[11pt]{article}
\usepackage{amsfonts}
\usepackage{graphics,amssymb,amsthm,amsmath,epsf,rotating}
\usepackage{graphicx}
\usepackage{subfigure}

\usepackage[round]{natbib}

\newtheorem{theorem}{Theorem}[section]

\newtheorem{prop}{Proposition}

\numberwithin{equation}{section}

\allowdisplaybreaks
\begin{document}

\date{}
\title{Exponentiated Extended Weibull-Power Series \\ Class of Distributions}
\author{ S. Tahmasebi$^1$ , A. A. Jafari$^{2,}$ \thanks{Corresponding:aajafari@yazd.ac.ir} \\
{\small $^1$Department of Statistics, Persian Gulf University, Bushehr, Iran}\\
{\small $^{2}$Department of Statistics, Yazd University, Yazd, Iran}}

\date{}
\maketitle
\begin{abstract}
In this paper, we introduce a new class of distributions by compounding the exponentiated extended Weibull family and power series family.
 This distribution contains several lifetime models such as the complementary extended Weibull-power series, generalized exponential-power series,
generalized linear failure rate-power series, exponentiated Weibull-power series, generalized modified Weibull-power series, generalized Gompertz-power series and exponentiated extended Weibull distributions as special cases. We obtain several properties of this new class of  distributions such as Shannon entropy, mean residual life, hazard  rate function, quantiles and moments. The maximum likelihood estimation procedure via a EM-algorithm is presented. 
\end{abstract}


\section{Introduction}

The extended Weibull (EW) family  contains various well-known distributions such as exponential, Pareto, Gompertz, Weibull, linear failure rate \citep{barlow-68}, modified Weibull \citep{la-xi-mu-03}, additive Weibull \citep{xi-ch-95,al-yu-13} and Chen \citep{chen-00new} distributions.
For more details see
\cite{na-ko-05}
and
\cite{ph-la-07}.

Using the given method by
\cite{gu-ku-99}, the EW family can be generalized. We call it exponentiated extended Weibull (EEW) distribution.
The cumulative distribution function (cdf) of this distribution is
\begin{eqnarray}\label{1e}
G(x;\alpha,\beta,\boldsymbol\Theta)=[1-e^{-\alpha H(x;\boldsymbol\Theta)}]^{\beta},\;\;\alpha>0,\;\; \;\beta>0,\;\;\;x\geq0,
\end{eqnarray}
and its probability density function (pdf) is
\begin{eqnarray}\label{fEEW}
{\rm g}(x;\alpha,\beta,\boldsymbol\Theta)=\alpha\beta h(x;\boldsymbol\Theta)e^{-\alpha H(x;\boldsymbol\Theta)}[1-e^{-\alpha H(x;\boldsymbol\Theta)}]^{\beta-1},
\end{eqnarray}
where $\boldsymbol\Theta$ is a vector of parameters, and $H(x;\boldsymbol\Theta)$ is a non-negative, continuous, monotone increasing, differentiable function of
$x$ such that $H(x;\boldsymbol \Theta)\rightarrow 0$ as $x \rightarrow 0^{+}$ and $H(x;\boldsymbol\Theta)\rightarrow \infty$ as $x \rightarrow \infty$.
It is denoted by ${\rm EEW}(\alpha,\beta,\boldsymbol\Theta)$.

The EEW distribuyion is a flexible  family and  contains many exponentiated distributions such as generalized exponential \citep{gu-ku-99}, exponentiated Weibull
\citep{mu-sr-93}, generalized Rayleigh  \citep{su-pa-01,ku-ra-05}, generalized modified Weibull \citep{ca-or-co-08}, generalized linear failure rate \citep{sa-ku-09}, and generalized Gompertz \citep{el-al-al-13} distributions.

In recent years, many distributions to model lifetime data have been introduced. The ba-
sic idea of introducing these models is that a lifetime of a system with $N$ (discrete random
variable) components and the positive continuous random variable, say $X_i$ (the lifetime of ith
omponent), can be denoted by the non-negative random variable $Y = \min(X_1, \dots ,X_N)$ or
$Y = \max(X_1, \dots ,X_N)$, based on whether the components are series or parallel.

In this paper, we compound the EEW family and power series distributions,  and  introduce a new class of distribution.
This class of distributions can be applied to reliability problems and its some properties are investigated in this paper.
We call it exponentiated extended Weibull-power series (EEWPS) class of distributions.  In similar way, some distributions are proposed in literature: The exponential-power series (EP) distribution by \cite{ch-ga-09},  Weibull-power series (WPS) distributions  by  \cite{mo-ba-11},
generalized exponential-power series (GEP)  distribution  by \cite{ma-ja-12},
complementary exponential power series by \cite{fl-bo-ca-13},
extended Weibull-power series (EWPS)  distribution by \cite{si-bo-di-co-13},
double bounded Kumaraswamy-power series by \cite{bi-ne-13},
Burr-power series by \cite{si-co-13},
generalized linear failure rate-power series (GLFRP)  distribution  by \cite{al-sh-14},
Birnbaum-Saunders-power series distribution by \cite{bo-si-co-14},
linear failure rate-power series by \cite{ma-ja-14}, and complementary extended Weibull-power series by \cite{co-si-14}.
Similar procedures are used by \cite{ro-lo-ca-le-12}, \cite{lu-sh-11},  \cite{na-co-or-14}  and \cite{lo-ma-ro-14}.
For compounding continuous distributions with discrete distributions,
 \cite{na-po-ri-13}
 introduced the package\textsf{ Compounding} in R software \citep{rdev-14}.

We provide three motivations for the EEWPS class of distributions, which can be applied in
some interesting situations as follows:
(i) This new class of distributions due to the stochastic representation $Y = \max(X_1,\dots,X_N)$,
can arises in parallel systems with identical components, where each component has the EEW
distribution lifetime. This model appears in many industrial applications and biological organisms
which the lifetime of the event is only the maximum ordered lifetime value among
all causes.
(ii)  The EEWPS class of distributions gives a
reasonable parametric fit to some modeling phenomenon with non-monotone hazard rates such as the
bathtub-shaped, unimodal and increasing-decreasing-increasing hazard rates, which are common
in reliability and biological studies. (iii) The time to the last failure can be appropriately
modeled by the EEWPS class of distributions.

The remainder of this paper is organized as follows:  The pdf and failure rate function of the new class of distributions are given in Section \ref{sec.dis}.
The special cases of the EEWPS distribution are considered in Section \ref{sec.spe}.
Some properties such as quantiles, moments, order statistics, Shannon entropy and mean residual life  are given in Section \ref{sec.pro}.
Estimation of parameters by maximum likelihood are discussed in Section \ref{sec.est}. Application to a real data set is presented  in Section \ref{sec.ex}.

\section{Introducing new family }
\label{sec.dis}

A discrete random variable, $N$  is a member of power series distributions (truncated at zero) if its probability mass function (pmf) is given by
\begin{equation}\label{fps}
p_{n}=P(N=n)=\frac{a_{n}\lambda^{n}}{C(\lambda)},\;\; n=1,2,\dots,
\end{equation}
where $a_{n}\geq0$, $C(\lambda)=\sum\limits_{n=1}^{\infty}a_{n}\lambda^{n}$, and $\lambda\in(0,s)$ is chosen in a way such that $C(\lambda)$ is finite and its first, second and third derivatives are defined and shown by $C'(.)$, $C''(.)$ and $C'''(.)$, respectively. The term ``power series distribution" is generally credited to
\cite{noack-50}. This family of distributions includes many of the most common distributions, including the binomial, Poisson, geometric, negative binomial, logarithmic distributions. For more details about power series distributions, see
\cite{jo-ke-ko-05},
 page 75.

\begin{theorem}
Let $N$ be a random variable denoting the number of failure causes which it is a member of power series distributions with pmf in \eqref{fps}. Also,  For given $N$, let $X_{1},X_{2},...,X_{N}$ be independent identically distributed random variables from EEW distribution with pdf in \eqref{fEEW}. Then $X_{(N)}=\max_{1\leq i\leq N}\{X_{i}\}$ has EEWPS class of distributions is denoted by ${\rm EEWPS}(\alpha,\beta,\lambda,$ $\boldsymbol\Theta)$ and has the following pdf:
\begin{equation}\label{fEEWPS}
f(x)=\alpha\beta\lambda h(x;\boldsymbol\Theta)e^{-\alpha H(x;\boldsymbol\Theta)}(1-e^{-\alpha H(x;\boldsymbol\Theta)})^{\beta-1}
\frac{C'\left(\lambda(1-e^{-\alpha H(x;\boldsymbol\Theta)})^{\beta}\right)}{C(\lambda)}, \;\;x>0.
\end{equation}
\end{theorem}

\begin{proof}
The conditional cdf of $X_{(N)}\mid N=n$  has ${\rm EEW}(\alpha,n\beta,\boldsymbol\Theta)$.
Hence,
\begin{equation}\label{eq.PXN}
P(X_{(N)}\leq x,N=n)=\frac{a_{n}\lambda^{n}}{C(\lambda)}[1-e^{- \alpha H(x;\boldsymbol\Theta)}]^{n\beta},
\end{equation}
and the marginal cdf of $X_{(N)}$ is
\begin{equation}\label{FEEWPS}
F(x)=\frac{C(\lambda (1-e^{-\alpha H(x;\boldsymbol\Theta)})^{\beta})}{C(\lambda)}, \;\; x>0.
\end{equation}
The derivative of $F$ with respect to $x$ is \eqref{fEEWPS}. Therefore,  $X_{(N)}$ has EEWPS distribution.
\end{proof}

\begin{prop}
 The pdf of EEWPS class can be expressed as infinite linear combination of density of order distribution, i.e. it can be written as
\begin{eqnarray}\label{eq.fpngn}
f(x)=\sum\limits_{n=1}^{\infty} p_{n} g_{(n)}(x;\alpha,n\beta,\boldsymbol\Theta),
\end{eqnarray}
where $g_{(n)}(x;\alpha,n\beta,\boldsymbol\Theta)$ is the pdf of 
EEW distribution with parameters $\alpha$, $n\beta$ and $\boldsymbol\Theta$.
\end{prop}
\begin{proof}
Consider $t=1-e^{-\alpha H(x;\boldsymbol\Theta)}$. So
\begin{eqnarray*}
f(x)&=&\alpha\beta\lambda h(x;\boldsymbol\Theta)e^{-\alpha H(x;\boldsymbol\Theta)}t^{\beta-1}
\frac{C'\left(\lambda t^{\beta}\right)}{C(\lambda)}\\
&=&\alpha\beta\lambda h(x;\boldsymbol\Theta)e^{-\alpha H(x;\boldsymbol\Theta)}t^{\beta-1}\frac{\sum\limits_{n=1}^{\infty}n a_{n}(\lambda t^{\beta})^{n-1}}{C(\lambda)}\\
&=&\sum\limits_{n=1}^{\infty}\frac{a_{n}\lambda^{n}}{C(\lambda)}n\alpha\beta h(x;\boldsymbol\Theta)e^{-\alpha H(x;\boldsymbol\Theta)} t^{n\beta-1}\\
&=&\sum\limits_{n=1}^{\infty} p_{n} g_{(n)}(x;\alpha,n\beta,\boldsymbol\Theta).
\end{eqnarray*}
\end{proof}

\begin{prop}\label{prop.1}
The limiting distribution of ${\rm EEWPS}(\beta,\lambda,\boldsymbol\Theta)$ when $\lambda\rightarrow 0^{+}$ is
\[ {\mathop{\lim }_{\lambda\rightarrow 0^{+}} F(x)}=[1-e^{-\alpha H(x;\boldsymbol\Theta)}]^{c\beta},\]
which is a EEW distribution with parameters $\alpha$, $c\beta$ and $\boldsymbol\Theta$, where $c=\min\{n\in {\mathbb N}: a_{n}>0\}$.
\end{prop}

\begin{proof}
Consider $t=1-e^{-\alpha H(x;\boldsymbol\Theta)}$. So
\begin{eqnarray*}
 {\mathop{\lim }_{\lambda\rightarrow 0^{+}} F(x)}&=&\mathop{\lim }_{\lambda\rightarrow 0^{+}}\frac{C(\lambda t^\beta)}{C(\lambda)}={\mathop{\lim }_{\lambda\rightarrow 0^{+}}\frac{\sum\limits_{n=1}^{\infty}a_{n}\lambda^{n}t^{n\beta}}{\sum\limits_{n=1}^{\infty}a_{n}\lambda^{n}}} \\
&=&{\mathop{\lim }_{\lambda\rightarrow 0^{+}}\frac{a_{c}t^{c\beta}+\sum\limits_{n=c+1}^{\infty}a_{n}\lambda^{n-c}t^{n\beta}}{a_{c}+\sum\limits_{n=c+1}^{\infty}a_{n}\lambda^{n-c}}}
=t^{c\beta}.
\end{eqnarray*}
\end{proof}

\begin{prop}
The hazard rate function of the EEWPS class of distributions is given by
\begin{eqnarray}\label{hGP}
 r(x)= \frac{\alpha\lambda\beta h(x;\boldsymbol\Theta)(1-t) t^{\beta-1}C'\left(\lambda t^{\beta}\right)}{C(\lambda)-C(\lambda t^{\beta})},
\end{eqnarray}
where $t=1-e^{-\alpha H(x;\boldsymbol\Theta)}$.
\end{prop}

\begin{proof}
Using \eqref{fEEWPS}, \eqref{FEEWPS} and definition of  hazard rate function as $r(x)=f(x)/(1-F(x)$, the proof is obvious.
\end{proof}

\section{Special cases}
\label{sec.spe}
In this Section, we consider some special cases of the EEWPS distribution.

\subsection{Complementary extended Weibull power series}
If $\beta=1$, then the pdf in \eqref{fEEWPS} becomes to
\begin{equation}\label{fCEWPS}
f(x)=\alpha\lambda h(x;\boldsymbol\Theta)e^{-\alpha H(x;\boldsymbol\Theta)}\frac{C'\left(\lambda(1-e^{-\alpha H(x;\boldsymbol\Theta)})\right)}{C(\lambda)}, \;\;x>0,
\end{equation}
which is the pdf of complementary extended Weibull power series (CEWPS) class of distributions introduced by \cite{co-si-14}.

\subsection{Generalized exponential-power series}
If $H(x;\boldsymbol\Theta)=x$, then the pdf in \eqref{fEEWPS} becomes to
\begin{equation}\label{fGEPS}
f(x)=\alpha\beta\lambda e^{-\alpha x}(1-e^{-\alpha x})^{\beta-1}
\frac{C'\left(\lambda(1-e^{-\alpha x})^{\beta}\right)}{C(\lambda)}, \;\;x>0.
\end{equation}
which is the pdf of generalized exponential-power series (GEPS) class of distributions  introduced by \cite{ma-ja-12}.
The GEPS class of distributions contains
complementary exponentiated exponential-geometric distribution introduced by
\cite{lo-ma-ca-13},
complementary exponential-geometric distribution introduced by
\cite{lo-ro-ca-11},
Poisson-exponential  distribution introduced by
\cite{ca-lo-fr-ba-11} and \cite{lo-ca-ba-11},
 complementary exponential -power series class of distributions introduced by \cite{fl-bo-ca-13}, generalized exponential distribution introduced by \cite{gu-ku-99}   and generalized exponential-geometric distribution introduced by \cite{bi-be-to-13} .

\subsection{Generalized linear failure rate-power series}
If $H(x;\boldsymbol\Theta)=\frac{ax}{\alpha}+\frac{bx^2}{2\alpha}$, then the pdf in \eqref{fEEWPS} becomes to
\begin{equation}\label{fGLFRPS}
f(x)=\beta\lambda (a+bx)e^{-ax-\frac{bx^2}{2}}(1-e^{-ax-\frac{bx^2}{2}})^{\beta-1}
\frac{C'\left(\lambda(1-e^{-ax-\frac{bx^2}{2}})^{\beta}\right)}{C(\lambda)}, \;\;x>0.
\end{equation}
which is the pdf of generalized linear failure rate-power series (GLFRPS) class of distributions introduced by \cite{al-sh-14}. It is a modification of generalized linear failure rate distribution introduced by \cite{sa-ku-09} and
generalized linear failure rate-geometric distribution introduced by \cite{na-sh-re-14}. If $b=0$, it becomes to GEPS class of distributions. Also, If  $\beta=1$, it becomes to linear failure rate-power series introduced by \cite{ma-ja-14}.

\subsection{Exponentiated Weibull-power series}
If $H(x;\boldsymbol\Theta)=x^\gamma$, then the pdf in \eqref{fEEWPS} becomes to
\begin{equation}\label{fEWPS}
f(x)=\alpha\beta\lambda \gamma x^{\gamma-1} e^{-\alpha x^\gamma}(1-e^{-\alpha x^\gamma})^{\beta-1}
\frac{C'\left(\lambda(1-e^{-\alpha x^\gamma})^{\beta}\right)}{C(\lambda)}, \;\;x>0.
\end{equation}
which is the pdf of exponentiated Weibull-power series (EWPS) class of distributions introduced by \cite{ma-sh-12}. It is a modification of exponentiated Weibull distribution introduced by \cite{mu-sr-93}.It is contain the complementary Weibull geometric distribution introduced by \cite{to-lo-ro-bo-14}. Also, the Marshall-Olkin extended Weibull distribution introduced by \cite{co-le-13} is a special case of EWPS.

\subsection{Generalized modified Weibull-power series}
If $H(x;\boldsymbol\Theta)=x^\gamma \exp(\tau x)$, then the pdf in \eqref{fEEWPS} becomes to
\begin{equation}\label{fGMWPS}
f(x)=\alpha\beta\lambda x^{\gamma-1}(\gamma+\tau x)e^{\tau x-\alpha x^\gamma \exp(\tau x)}
\frac{C'\left(\lambda(1-e^{-\alpha x^\gamma \exp(\tau x)})^{\beta}\right)}{(1-e^{-\alpha x^\gamma \exp(\tau x)})^{1-\beta}C(\lambda)}, \;\;x>0,
\end{equation}
and we call generalized modified Weibull-power series (GMWPS) class of distributions. It is contained the generalized modified Weibull distribution introduced by \cite{ca-or-co-08}.
If $\tau=0$, then GMWPS class of distributions becomes to EWPS class of distributions.

%

\subsection{Generalized Gompertz-power series}
If $H(x;\boldsymbol\Theta)=\frac{1}{\gamma}(e^{\gamma x}-1)$, then the pdf in \eqref{fEEWPS} becomes to
\begin{equation}\label{fGPPS}
f(x)=\alpha\beta\lambda e^{\gamma x}e^{- \frac{\alpha}{\gamma}(e^{\gamma x}-1)}(1-e^{- \frac{\alpha}{\gamma}(e^{\gamma x}-1)})^{\beta-1}
\frac{C'\left(\lambda(1-e^{- \frac{\alpha}{\gamma}(e^{\gamma x}-1)})^{\beta}\right)}{C(\lambda)}, \;\;x>0.
\end{equation}
and we call generalized Gompertz-power series  class of distributions. It is contained the generalized Gompertz distribution introduced by \cite{el-al-al-13}.

%
%
%

\section{Statistical properties}
\label{sec.pro}
In this section, some properties of EEWPS class of distributions  such as quantiles, moments, order statistics, Shannon entropy and mean residual life  are derived. Using \eqref{eq.fpngn}, we can obtain
\begin{eqnarray}
F(x)=\sum\limits_{n=1}^{\infty} p_{n} G_{(n)}(x;\alpha, n\beta,\boldsymbol\Theta)=\sum\limits_{n=1}^{\infty} p_{n}t^{n\beta},
\end{eqnarray}
where $t=1-e^{-\alpha H(x;\boldsymbol\Theta)}$. Based on the mathematical quantities of the baseline pdf $g_{(n)}(x;\alpha,n\beta,$ $\boldsymbol\Theta)$, we can obtain some statistical quantities such as
ordinary and incomplete moments, generating function and mean deviations of this family of distributions.

\subsection{Quantiles and Moments}
Let
$$
X=G^{-1}\left(\frac{C^{-1}\left(C(\lambda)U\right)}{\lambda}\right),
$$
where  $U$ has a uniform  distribution on $(0, 1)$, $G^{-1}(y)=H^{-1}[-\frac{1}{\alpha}\ln(1-y^{\frac{1}{\beta}})]$
 and $C^{-1}(.)$ is the inverse function of $C(.)$. Then $X$ has the ${\rm EEWPS}(\alpha,\beta,\lambda,\boldsymbol\Theta)$ distribution.
This  result helps in simulating data from the EEWPS distribution with generating  uniform distribution data.

%
%

\begin{theorem}
Consider $X\sim {\rm EEWPS}(\alpha,\beta,\lambda,\boldsymbol\Theta)$. Then
the moment generating function of EEWPS is
\begin{eqnarray}
M_{X}(t)=\sum\limits_{n=1}^{\infty}\sum\limits_{j=0}^{\infty}p_{n}\binom{n\beta}{j+1}(-1)^{j} M_{Y}(t),
\end{eqnarray}
where $Y$ has ${\rm EEW}(\alpha(j+1),1,\boldsymbol \Theta)$.
\end{theorem}

\begin{proof}
The Laplace transform of the EEWPS class can be expressed as
\begin{eqnarray*}
L(s)=E(e^{-sX})=\sum\limits_{n=1}^{\infty} P(N=n)L_{n}(s),
\end{eqnarray*}
where $L_{n}(s)$ is the Laplace transform of EEW distribution with parameters $\alpha$, $n\beta$ and $\boldsymbol\Theta$ given as
\begin{eqnarray}
L_{n}(s)&=&\int_{0}^{+\infty} e^{-sx}n \alpha \beta h(x;\boldsymbol\Theta)e^{-\alpha H(x;\boldsymbol\Theta)}[1-e^{-\alpha H(x;\boldsymbol\Theta)}]^{n\beta-1}dx \nonumber\\
&=&n\alpha \beta \int_{0}^{+\infty} e^{-sx} h(x;\boldsymbol\Theta)\sum\limits_{j=0}^{\infty}\binom{n\beta-1}{j}(-1)^{j}e^{-(j+1)\alpha H(x;\boldsymbol\Theta)}dx \nonumber\\
&=&\sum\limits_{j=0}^{\infty}n\beta\binom{n\beta-1}{j}(-1)^{j}\int_{0}^{+\infty}\frac{\alpha(j+1)}{j+1}h(y;\boldsymbol\Theta)e^{-(j+1)\alpha H(y;\boldsymbol\Theta)-sy}dy\nonumber\\
&=&\sum\limits_{j=0}^{\infty}\binom{n\beta}{j+1}(-1)^{j}L_{1}(s), \nonumber
\end{eqnarray}
where $ L_{1}(s)$ is the Laplace transform of the ${\rm EEW}(\alpha(j+1),1,\boldsymbol \Theta)$.
Therefore,
the moment generating function of EEWPS is
\begin{eqnarray*}
M_{X}(t)
&=&\sum\limits_{n=1}^{\infty}p_{n}L_{n}(-t)\\
&=&\sum\limits_{n=1}^{\infty}\sum\limits_{j=0}^{\infty}p_{n}\binom{n\beta}{j+1}(-1)^{j}L_{1}(-t) \\
&=&\sum\limits_{n=1}^{\infty}\sum\limits_{j=0}^{\infty}p_{n}\binom{n\beta}{j+1}(-1)^{j} M_{Y}(t).
\end{eqnarray*}
\end{proof}

\begin{theorem}
The noncentral moment functions of EEWPS is
\begin{eqnarray}\label{eq.mur}
\mu_{r}
=\sum\limits_{n=1}^{\infty}\frac{a_{n}\lambda^{n}}{C(\lambda)}\sum\limits_{j=0}^{\infty}\binom{n\beta}{j+1}(-1)^{j}\mu'_{r}
=\sum\limits_{n=1}^{\infty}\sum\limits_{j=0}^{\infty}p_{n}\binom{n\beta}{j+1}(-1)^{j}\mu'_{r},
\end{eqnarray}
where $\mu'_{r}=E[Y^{r}]$ and $Y$ has ${\rm EEW}(\alpha(j+1),1,\boldsymbol \Theta)$.

\end{theorem}

\begin{proof}
We can use $M_{X}(t)$ to obtain $\mu_{r}$. But from the direct calculation, proof is obvious.
\end{proof}

With considering $H(x)=x^\gamma$ and $C(\lambda)=\lambda(1-\lambda)^{-1}$,
we calculated the first four moments with
different values of parameters for the EEWPS distribution using \eqref{eq.mur}. Also,  we computed these values from the direct definition by numerical integration. We found that the results are same.  The values are given in Tables \ref{table.mom}.

\begin{table}[ht]
\begin{center}
\caption{The four moments of EEWPS model.}\label{table.mom}
{\small
\begin{tabular}{|cccc|cccc|} \hline
$\alpha$ & $\beta$ & $\lambda$ & $\gamma$ & $\mu_1$ & $\mu_2$ & $\mu_3$ & $\mu_4$ \\ \hline
0.3 & 0.3 & 0.2 & 2.0 & 0.936 & 1.594 & 3.520 & 9.164 \\
0.3 & 0.3 & 0.2 & 5.0 & 0.856 & 0.884 & 1.011 & 1.237 \\
0.3 & 0.3 & 0.8 & 2.0 & 1.656 & 3.719 & 9.722 & 28.292 \\
0.3 & 0.3 & 0.8 & 5.0 & 1.150 & 1.446 & 1.916 & 2.631 \\
0.3 & 2.0 & 0.2 & 2.0 & 2.192 & 5.446 & 14.976 & 44.841 \\
0.3 & 2.0 & 0.2 & 5.0 & 1.345 & 1.853 & 2.606 & 3.734 \\
0.3 & 2.0 & 0.8 & 2.0 & 2.835 & 8.733 & 28.694 & 99.531 \\
0.3 & 2.0 & 0.8 & 5.0 & 1.500 & 2.285 & 3.530 & 5.521 \\
0.8 & 0.3 & 0.2 & 2.0 & 0.573 & 0.598 & 0.808 & 1.289 \\
0.8 & 0.3 & 0.2 & 5.0 & 0.704 & 0.597 & 0.561 & 0.565 \\
0.8 & 0.3 & 0.8 & 2.0 & 1.014 & 1.394 & 2.232 & 3.979 \\
0.8 & 0.3 & 0.8 & 5.0 & 0.945 & 0.977 & 1.064 & 1.201 \\
0.8 & 2.0 & 0.2 & 2.0 & 1.342 & 2.042 & 3.439 & 6.306 \\
0.8 & 2.0 & 0.2 & 5.0 & 1.106 & 1.252 & 1.446 & 1.704 \\
0.8 & 2.0 & 0.8 & 2.0 & 1.736 & 3.275 & 6.589 & 13.997 \\
0.8 & 2.0 & 0.8 & 5.0 & 1.233 & 1.543 & 1.960 & 2.519 \\
2.0 & 0.3 & 0.2 & 2.0 & 0.362 & 0.239 & 0.204 & 0.206 \\
2.0 & 0.3 & 0.2 & 5.0 & 0.586 & 0.414 & 0.324 & 0.271 \\
2.0 & 0.3 & 0.8 & 2.0 & 0.641 & 0.558 & 0.565 & 0.637 \\
2.0 & 0.3 & 0.8 & 5.0 & 0.787 & 0.677 & 0.614 & 0.577 \\
2.0 & 2.0 & 0.2 & 2.0 & 0.849 & 0.817 & 0.870 & 1.009 \\
2.0 & 2.0 & 0.2 & 5.0 & 0.921 & 0.867 & 0.835 & 0.819 \\
2.0 & 2.0 & 0.8 & 2.0 & 1.098 & 1.310 & 1.667 & 2.239 \\
2.0 & 2.0 & 0.8 & 5.0 & 1.026 & 1.070 & 1.131 & 1.210 \\ \hline
\end{tabular}
}

\end{center}
\end{table}

\subsection{Order statistic}

Let $X_{1},X_{2},\dots,X_{m}$  be a random sample of size $m$ from ${\rm EEWPS}(\alpha,\beta,\lambda,\boldsymbol\Theta)$, then the pdf of the $i$th order statistic, say $X_{i:m}$, is given by
\begin{eqnarray*}
f_{i:m}(x)&=&\frac{m!}{(i-1)!(m-i)!}f(x)\left[\frac{C(\lambda t^{\beta})}{C(\lambda)}\right]^{i-1}\left[1-\frac{C(\lambda t^{\beta})}{C(\lambda)}\right]^{m-i} \nonumber \\
&=&\frac{m!}{(i-1)!(m-i)!}f(x)\sum\limits_{j=0}^{m-i}\binom{m-i}{j}(-1)^{j}\left[\frac{C(\lambda t^{\beta})}{C(\lambda)}\right]^{j+i-1} \nonumber\\
&=&\frac{m!}{(i-1)!(m-i)!}\sum\limits_{n=1}^{\infty}\sum\limits_{j=0}^{m-i} p_{n} g_{(n)}(x;\alpha,n\beta,\boldsymbol\Theta)\binom{m-i}{j}(-1)^{j}\left[\frac{C(\lambda t^{\beta})}{C(\lambda)}\right]^{j+i-1} \nonumber\\
&=&\frac{m!}{(i-1)!(m-i)!}\sum\limits_{n=1}^{\infty}\sum\limits_{j=0}^{m-i}w_{j} p_{n} g_{(n)}(x;\alpha,n\beta,\boldsymbol\Theta)\left[\frac{C(\lambda t^{\beta})}{C(\lambda)}\right]^{j+i-1},
\end{eqnarray*}
where $f$ is the pdf of EEWP class of distributions, $t=1-e^{-\alpha H(x;\boldsymbol\Theta)}$ and $w_{j}=\binom{m-i}{j}(-1)^{j}$. Also, the cdf of $X_{i:m}$ is given by
$$
F_{i:m}(x)=\sum\limits_{k=i}^{m}\sum\limits_{j=0}^{m-k}(-1)^{j}\binom{m-k}{j}\binom{m}{k}\left[\frac{C(\lambda t^{\beta})}{C(\lambda)}\right]^{j+k}.
$$

An analytical expression for  $r$th moment of  order statistics $X_{i:m}$ is obtained as
\begin{eqnarray*}
E[X_{i:m}^{r}]= \frac{m!}{(i-1)!(m-i)!}\sum\limits_{n=1}^{\infty}\sum\limits_{j=0}^{m-i}w_{j} p_{n}E[Z^{r}(F(Z))^{j+i-1}],
\end{eqnarray*}
where $Z$ has a EEW distribution with parameters $\alpha$, $n\beta$ and $\boldsymbol\Theta$.


\subsection{Shannon entropy and mean residual life }

The maximum entropy method is a powerful technique in the field of probability and statistics. It is introduced by
\cite{jaynes-57}
and closely related to the Shannon's entropy.
Also, it  is applied in a wide variety of
fields and  used for the characterization of pdf's; see, for example,
\cite{kapur-94}
\cite{soofi-00}
 and
\cite{zo-ba-09}.
\cite{sh-jo-80}
treated the maximum entropy method axiomatically.

Considers a class of pdf's
\begin{equation}\label{eq.Fclass}
F=\{f(x;\alpha,\beta,\lambda,\boldsymbol\Theta): E_{f}(T_{i}(X))=\beta_{i}, \; i=0,1,....,m\},
\end{equation}
where $T_{1}(X),..., T_{m}(X)$ are absolutely integrable functions with respect to $f$, and $T_{0}(X)=1$.
Also, consider the shannon's entropy of none-negative continuous random variable $X$ with pdf $f$  defined  by
\cite{shan-48}
as
\begin{equation}\label{eq.sh}
H_{sh}(f)=E[-\log f(X)]=-\int_{0}^{+\infty} f(x)\log (f(x))dx.
\end{equation}
The maximum entropy distribution is the pdf of the class $F$, denoted by $f^{ME}$ determined as
\begin{equation*}
f^{ME}(x;\lambda,\beta,\boldsymbol\Theta)=\arg \max_{f\in F} H_{sh}(f).
\end{equation*}



Now,  suitable constraints are derived in order to provide a
maximum entropy characterization for the class \eqref{eq.Fclass} based on \cite{jaynes-57}.
For this purpose, the next result plays an important role.

\begin{prop}
Let $X$  has ${\rm EEWPS}(\alpha,\beta,\lambda,\boldsymbol \Theta)$ with the pdf given by \eqref{fEEWPS}. Then,\\
\noindent i.
\begin{eqnarray*}
E\left[\log(C'(\lambda(1-e^{-\alpha H(X;\boldsymbol\Theta)})^{\beta}))\right]&=&\frac{\lambda}{C(\lambda)}E\left[C'(\lambda(1-e^{-\alpha H(Y;\boldsymbol\Theta)})^
{\beta})\right. \\
&&\times
\left.\log(C'(\lambda(1-e^{-\alpha H(Y;\boldsymbol\Theta)})^{\beta}))\right],
\end{eqnarray*}
\noindent ii.
$$
E\left[\log(h(X;\boldsymbol\Theta))\right]=\frac{\lambda}{C(\lambda)}
E\left[C'(\lambda(1-e^{-H(Y;\boldsymbol\Theta)})^{\beta})\log(h(Y;\boldsymbol\Theta))\right],
$$
\noindent iii.
$$
E\left[\log(1-e^{-\alpha H(X;\boldsymbol\Theta)})\right]=\frac{\lambda}{C(\lambda)}E\left[C'(\lambda(1-e^{-\alpha H(Y;\boldsymbol\Theta)})^{\beta})\log(1-e^{-\alpha H(Y;\boldsymbol\Theta)})\right],
$$
where $Y$ follows the EEW distribution with the pdf in \eqref{fEEW}.
\end{prop}

An explicit expression of Shannon entropy for EEWPS  distribution is obtained as
\begin{eqnarray}
H_{sh}(f)&=&-\log(\alpha\beta\lambda)-\frac{\lambda}{C(\lambda)}E[C'(\lambda(1-e^{-H(Y;\boldsymbol\Theta)})^{\beta})\log(C'(\lambda(1-e^{-H(Y;\boldsymbol\Theta)})^{\beta}))]
\nonumber\\
&&+\log[C(\lambda)]-(\beta-1)\frac{\lambda}{C(\lambda)}E[C'(\lambda(1-e^{-H(Y;\boldsymbol\Theta)})^{\beta})\log(1-e^{-H(Y;\boldsymbol\Theta)})]\nonumber\\
&&-\frac{\lambda}{C(\lambda)}E[C'(\lambda(1-e^{-H(Y;\boldsymbol\Theta)})^{\beta})\log(h(Y;\boldsymbol\Theta))].
\end{eqnarray}

Also,  the mean residual life function of $X$ is given by
\begin{eqnarray}
m(t)=E[X-t|X>t]&=&\frac{\int_{t}^{+\infty}(x-t)f(x)dx}{1-F(t)}=\frac{C(\lambda)\sum\limits_{n=1}^{\infty} p_{n} \int_{t}^{+\infty}z g_{(n)}(z;\alpha,n\beta,\boldsymbol\Theta)dz}{C(\lambda)-C(\lambda G(x))}-t \nonumber\\
&=&\frac{C(\lambda)\sum\limits_{n=1}^{\infty} p_{n} E[ZI_{(Z>t)}]}{C(\lambda)-C(\lambda G(x))}-t,
\end{eqnarray}
where $Z$ has a EEW distribution with parameters $\alpha$, $n\beta$ and $\boldsymbol\Theta$.


\subsection{Reliability and average lifetime}
In the context of reliability, the stress - strength model describes the life of a component which has a random strength
$X$  subjected to a random  stress  $Y$. The component  fails at the instant that the  stress  applied to it  exceeds the  strength, and
the  component  will function satisfactorily whenever $X > Y$. Hence, $R=P(X > Y)$  is a measure of component reliability.
It has many applications especially in engineering concept. Here, we obtain the form for the reliability $R$ when $X$ and $Y$ are
independent random variables having the same EEWPS distribution. The quantity $R$ can be expressed as
\begin{eqnarray}
R&=&\int_{0}^{\infty}f(x;\alpha,\beta,\lambda,\boldsymbol\Theta)F(x;\alpha,\beta,\lambda,\boldsymbol\Theta)dx=\int_{0}^{\infty}\lambda g(x)\frac{C'(\lambda G(x))C(\lambda G(x))}{C^{2}(\lambda)}dx \nonumber\\
&=&\sum\limits_{n=1}^{\infty} p_{n} \int_{0}^{\infty}g_{(n)}(x;\alpha,n\beta,\boldsymbol\Theta)\frac{C(\lambda G(x))}{C(\lambda)}\ dx.
\end{eqnarray}

\section{Estimation}
\label{sec.est}
In this section, we first study the maximum likelihood estimations (MLE's) of the parameters. Then, we propose an  Expectation-Maximization (EM) algorithm to estimate the parameters.

\subsection{The MLE's}

 Let $x_{1},\dots, x_{n}$ be observed value from the EEWPS distribution with parameters ${\boldsymbol \xi}=(\alpha, \beta,\lambda,$ $\boldsymbol\Theta)^T$.
The  log-likelihood function is given by
\begin{eqnarray*}
 l_{n} = l_{n}(\boldsymbol\xi;\boldsymbol{x})&=&n[\log(\alpha)+\log(\beta)+\log(\lambda)-\log(C(\lambda))]+\sum\limits_{i=1}^{n}\log[h(x_{i};\boldsymbol\Theta)]\nonumber\\
 &&-\alpha\sum_{i=1}^n H(x_{i};\boldsymbol\Theta)+(\beta-1)\sum\limits_{i=1}^{n}\log t_{i}+ \sum\limits_{i=1}^{n}\log(C'(\lambda t_{i}^{\beta})),
\end{eqnarray*}
where $\boldsymbol{x}=(x_{1},\dots, x_{n})$ and  $t_{i}=1-e^{-\alpha H(x_{i};\boldsymbol\Theta)}$.
The components of the score function $U({\boldsymbol\xi};{\boldsymbol x})=(\frac{\partial l_{n}}{\partial\alpha},\frac{\partial l_{n}}{\partial\beta},\frac{\partial l_{n}}{\partial\lambda},\frac{\partial l_{n}}{\partial\boldsymbol\Theta})^{T}$ are
\begin{eqnarray}
\frac{\partial l_{n}}{\partial \alpha}&=& \frac{n}{\alpha}-\sum_{i=1}^n H(x_{i};\boldsymbol\Theta)\label{eq.alpha},\\
\frac{\partial l_{n}}{\partial \beta}&=&\frac{n}{\beta}+\sum\limits_{i=1}^{n}\log(t_{i})+\sum\limits_{i=1}^{n}\frac{\lambda t_{i}^{\beta}\log(t_{i})C''(\lambda t_{i}^{\beta})}{C'(\lambda t_{i}^{\beta})}, \label{eq.beta}\\
\frac{\partial l_{n}}{\partial \lambda}&=&\frac{n}{\lambda}-\frac{n C'(\lambda)}{C(\lambda)}+ \sum\limits_{i=1}^{n}\frac{t_{i}^{\beta}C''(\lambda t_{i}^{\beta})}{C'(\lambda t_{i}^{\beta})}, \label{eq.lambda}\\
\frac{\partial l_{n}}{\partial \Theta_{k}}&=&\sum\limits_{i=1}^{n}\frac{\partial h(x_{i};\boldsymbol\Theta)}{\partial \Theta_{k}}.\frac{1}{h(x_{i};\boldsymbol\Theta)}-\alpha\sum_{i=1}^n \frac{\partial H(x_{i};\boldsymbol\Theta)}{\partial \Theta_k}\nonumber\\
&&+(\beta-1)\sum\limits_{i=1}^{n}\frac{\frac{\partial t_{i}}{\partial \Theta_{k}}}{t_{i}}+\beta \lambda\sum\limits_{i=1}^{n}\frac{[\frac{\partial t_{i}}{\partial \Theta_{k}}] t_{i}^{\beta-1}C''(\lambda t_{i}^{\beta})}{C'(\lambda t_{i}^{\beta})}, \label{eq.theta}
 \end{eqnarray}
 where $\Theta_k$ is the $k$th element of the vector $\boldsymbol\Theta$.

The MLE of ${\boldsymbol \xi}$, say $\hat{\boldsymbol\xi}$, is obtained by solving the nonlinear system $U({\boldsymbol\xi};{\boldsymbol x})={\boldsymbol 0}$. We cannot get an explicit form for this nonlinear system of equations and they can be calculated by using a numerical method, like the Newton method or the bisection method. Only,
for given $\boldsymbol \Theta$, from \eqref{eq.alpha} we have
$$
\alpha=\frac{n}{\sum_{i=1}^n H(x_{i};\boldsymbol\Theta)}.
$$
Therefore, \eqref{eq.theta} becomes
\begin{eqnarray}
&&\sum\limits_{i=1}^{n}\frac{\partial h(x_{i};\boldsymbol\Theta)}{\partial \Theta_{k}}.\frac{1}{h(x_{i};\boldsymbol\Theta)}-\frac{n}{\sum_{i=1}^n H(x_{i};\boldsymbol\Theta)}\sum_{i=1}^n \frac{\partial H(x_{i};\boldsymbol\Theta)}{\partial \Theta_k}\nonumber\\
&&+(\beta-1)\sum\limits_{i=1}^{n}\frac{\frac{\partial t_{i}}{\partial \Theta_{k}}}{t_{i}}+\beta \lambda\sum\limits_{i=1}^{n}\frac{[\frac{\partial t_{i}}{\partial \Theta_{k}}] t_{i}^{\beta-1}C''(\lambda t_{i}^{\beta})}{C'(\lambda t_{i}^{\beta})}. \label{eq.theta2}
\end{eqnarray}

\begin{theorem}
The pdf, $f(x|{\boldsymbol \Theta})$, of EEWPS distribution satisfies on the regularity condistions, i.e.
\begin{itemize}
\item[i.] the support of $f(x|{\boldsymbol \Theta})$ does not depend on  ${\boldsymbol \Theta}$,
\item[ii.] $f(x|{\boldsymbol \Theta})$ is twice continuously differentiable with respect to ${\boldsymbol \Theta}$,
\item[iii.] the differentiation and integration are interchangeable in the sense that
\end{itemize}
$$
\frac{\partial}{\partial {\boldsymbol \Theta}} \int_{-\infty}^{\infty}f(x|{\boldsymbol \Theta})dx= \int_{-\infty}^{\infty} \frac{\partial}{\partial {\boldsymbol \Theta}}f(x|{\boldsymbol \Theta})dx, \ \ \ \frac{\partial^2}{\partial {\boldsymbol \Theta}\partial{\boldsymbol \Theta}^T} \int_{-\infty}^{\infty}f(x|{\boldsymbol \Theta})dx= \int_{-\infty}^{\infty} \frac{\partial^2}{\partial {\boldsymbol \Theta}\partial{\boldsymbol \Theta}^T}f(x|{\boldsymbol \Theta})dx.
$$

\end{theorem}
\begin{proof}
The proof is obvious and for more details, see \cite{ca-be-01} Section 10.
\end{proof}

The asymptotic confidence intervals of these parameters will be derived based on  Fisher information matrix.
It is well-known that under regularity conditions, the asymptotic distribution of
$\sqrt{n}(\hat{\boldsymbol\xi}-{\boldsymbol\xi})$ is multivariate normal with mean ${\boldsymbol 0}$ and variance-covariance matrix $J_n^{-1}({\boldsymbol\xi})$, 
where $J_n({\boldsymbol\xi})=\lim_{n\rightarrow \infty} I_n({\boldsymbol\xi})$, and $I_n({\boldsymbol\xi})$ is the observed information matrix as
\[I_n\left(\boldsymbol\xi\right)=-\left[ \begin{array}{ccccc}
U_{\alpha \alpha } &U_{\alpha \beta } & U_{\alpha\lambda } & | & U^{T}_{\alpha \boldsymbol\Theta} \\
U_{\alpha \beta } &U_{\beta \beta } & U_{\beta\lambda } & | & U^{T}_{\beta \boldsymbol\Theta} \\
U_{\alpha \beta }&U_{\lambda \beta } & U_{\lambda \lambda } &| &U^{T}_{\lambda \boldsymbol\Theta } \\
-&- & - &- &- \\
U_{\alpha \boldsymbol\Theta } &U_{\beta \boldsymbol\Theta } & U_{\lambda \boldsymbol\Theta } & | & U_{\boldsymbol\Theta\boldsymbol\Theta}\\
\end{array} \right],\]
whose elements are obtained by derivative the equations \eqref{eq.alpha}-\eqref{eq.theta} with respect to parameters.

\subsection{ EM-algorithm}

The traditional methods to obtain the MLE's are numerical methods for solving the equations \eqref{eq.alpha}-\eqref{eq.theta}, and sensitive to the initial values. Therefore, we develop an EM algorithm for obtaining the MLE's of the  parameters of EEWPS class of distributions. It
is a very powerful tool in handling the incomplete data problem
\citep{de-la-ru-77}.
It is an iterative method, and  there are two steps in each iteration: Expectation step or the E-step and the Maximization step or the M-step. The EM algorithm is especially useful if the complete data set is easy to analyze.

Using \eqref{eq.PXN}, we define a hypothetical complete-data distribution with a joint pdf in the form
$$
g(x,z;{\boldsymbol\xi})=\frac{a_{z}\lambda^{z}}{C(\lambda)}z\alpha\beta h(x;\boldsymbol\Theta) (1-t)t^{z\beta-1}, \ \ \ x>0, \ \ \ z\in {\mathbb N},
$$
where $t=1-e^{-\alpha H(x;\boldsymbol\Theta)}$.
 The E-step of an EM cycle requires the expectation of $(Z|X;{\boldsymbol\xi}^{(r)})$ where ${\boldsymbol\xi}^{(r)}=(\alpha^{(r)},\beta^{(r)},\lambda^{(r)},\boldsymbol{\Theta}^{(r)})$ is the current estimate (in the $r$th iteration) of ${\boldsymbol\xi}$.
The expected value of $Z|X=x$ is
\begin{eqnarray}
E(Z|X=x)=
1+\frac{\lambda t^{\alpha}C''(\lambda t^{\beta})}{C'(\lambda t^{\beta})}.
\end{eqnarray}

The M-step of EM cycle is completed by using the MLE over $\boldsymbol\Theta$, with the missing $z$'s replaced by their conditional expectations given
above.
Therefore, the log-likelihood for the complete-data ${\boldsymbol y}=(x_1,\dots,x_n, z_1,...,z_n)$  is
\begin{eqnarray} \label{eq.ls}
l^{\ast}(\boldsymbol {y};\boldsymbol\xi)&\propto& \sum\limits_{i=1}^{n}z_{i}\log(\lambda)
+n\log(\alpha\beta)
+\sum\limits_{i=1}^{n}\log h(x_{i};\boldsymbol{\Theta})
-\alpha\sum\limits_{i=1}^{n}H(x_{i};\boldsymbol{\Theta})\nonumber\\
&&+\sum\limits_{i=1}^{n}(z_{i}\beta-1)\log(1-e^{-\alpha H(x_{i};\boldsymbol\Theta)})
-n\log(C(\lambda)).
\end{eqnarray}
On differentiation of (\ref{eq.ls}) with respect to parameters $\alpha$, $\beta$, $\lambda$ and $\Theta_{k}$, we obtain the components of the score function
as
\begin{eqnarray*}
\frac{\partial l^{\ast}_{n}}{\partial\alpha}&=&\frac{n}{\alpha}-\sum\limits_{i=1}^{n}H(x_{i};\boldsymbol{\Theta})+\sum\limits_{i=1}^{n}(z_{i}\beta-1)\frac{H(x_{i};\boldsymbol{\Theta})e^{-\alpha H(x_{i};\boldsymbol\Theta)}}{1-e^{-\alpha H(x_{i};\boldsymbol\Theta)}},\\
\frac{\partial l^{\ast}_{n}}{\partial\beta}&=&\frac{n}{\beta}+\sum\limits_{i=1}^{n}z_{i}\log(1-e^{-\alpha H(x_{i};\boldsymbol\Theta)}),\\
\frac{\partial l^{\ast}_{n}}{\partial\lambda}&=&\sum\limits_{i=1}^{n}\frac{z_{i}}{\lambda}-n\frac{C'(\lambda)}{C(\lambda)},\\
\frac{\partial l^{\ast}_{n}}{\partial\Theta_{k}}&=&\sum\limits_{i=1}^{n}\frac{\partial h(x_{i};\boldsymbol\Theta)}{\partial \Theta_{k}}.\frac{1}{h(x_{i};\boldsymbol\Theta)}-\alpha\sum\limits_{i=1}^{n}\frac{\partial H(x_{i};\boldsymbol\Theta)}{\partial \Theta_{k}}+\sum\limits_{i=1}^{n}(z_{i}\beta-1)\frac{\frac{\partial H(x_{i};\boldsymbol\Theta)}{\partial \Theta_{k}}}{1-e^{-\alpha H(x_{i};\boldsymbol\Theta)}}.
\end{eqnarray*}
Therefore, we obtain the iterative procedure of the EM-algorithm as
\begin{eqnarray*}
&&\hat{\beta}^{(j+1)}=\frac{-n}{\sum\limits_{i=1}^{n}\hat{z}_{i}^{(j)}\log[1-e^{-\hat{\alpha}^{(j)}H(x_{i};\hat{\boldsymbol \Theta}^{(j)})}]},  \\ &&\hat{\lambda}^{(j+1)}=\frac{C(\hat{\lambda}^{(j+1)})}{nC'(\hat{\lambda}^{(j+1)})}\sum\limits_{i=1}^{n}\hat{z}_{i}^{(j)},\\
&&\frac{n}{\hat{\alpha}^{(j+1)}}-\sum\limits_{i=1}^{n}H(x_{i};\hat{\boldsymbol\Theta}^{(j)})+\sum\limits_{i=1}^{n}(\hat{z_{i}}^{(j)}\hat{\beta}^{(j)}-1)
\frac{H(x_{i};\hat{\boldsymbol\Theta}^{(j)})e^{-\hat{\alpha}^{(j+1)} H(x_{i};\hat{\boldsymbol\Theta}^{(j)})}}{1-e^{-\hat{\alpha}^{(j+1)} H(x_{i};\hat{\boldsymbol\Theta}^{(j)})}}=0,
\\
&&\sum\limits_{i=1}^{n}\frac{\partial h(x_{i};\hat{\boldsymbol\Theta}^{(j+1)})}{\partial {\Theta_{k}}}.\frac{1}{h(x_{i};\hat{\boldsymbol\Theta}^{(j+1)})}-\hat{\alpha}^{(j)}\sum\limits_{i=1}^{n}\frac{\partial H(x_{i};\hat{\boldsymbol\Theta}^{(j+1)})}{\partial {\Theta_{k}}}\\
&& \qquad \qquad+\sum\limits_{i=1}^{n}\frac{\partial H(x_{i};\hat{\boldsymbol\Theta}^{(j+1)})}{\partial {\Theta_{k}}}.
\frac{\hat{z_{i}}^{(j)}\hat{\beta}^{(j)}-1}{1-e^{-\hat{\alpha}^{(j)} H(x_{i};\hat{\boldsymbol\Theta}^{(j+1)})}}=0,
\end{eqnarray*}
where $\hat{\lambda}^{(j+1)}$ , $\hat{\alpha}^{(j+1)}$ and $\hat{\Theta}_{k}^{(j+1)}$ are found numerically. Here,  we have
$$
\hat{z}_{i}^{(j)}=1+\frac{\lambda^{*(j)}C''(\lambda^{*(j)})}{C'(\lambda^{*(j)})}, \ \ i=1,2,...,n,
$$
where
$\lambda^{*(j)}=\hat{\lambda}^{(j)}[1-e^{-\hat{\alpha}^{(j)}H(x_{i};\hat{\Theta_{k}}^{(j)})}]^{\hat{\beta}^{(j)}}$.

We can use the results of
\cite{lou-82}
to obtain the standard errors of the estimators from the EM-algorithm. Consider
$\ell_{c}({\boldsymbol\Theta};{\boldsymbol x})=E(I_{c}({\boldsymbol\Theta};{\boldsymbol y})|{\boldsymbol x})$,
 where
$I_{c}({\boldsymbol\Theta};{\boldsymbol y})=-[\frac{\partial U({\boldsymbol y};{\boldsymbol\Theta})}{\partial{\boldsymbol\Theta}}]$
is the $(k+3)\times (k+3)$ observed information matrix. If
$
\ell_{m}({\boldsymbol\Theta};{\boldsymbol x})=Var[U({\boldsymbol y};{\boldsymbol\Theta})|{\boldsymbol x}]$, then, we obtain the observed information as
\begin{equation}\label{EM information}
J(\hat{\boldsymbol\Theta};{\boldsymbol x})={\ell }_c(\hat{ \boldsymbol\Theta};{\boldsymbol x})-{\ell }_m(\hat{\boldsymbol\Theta};{\boldsymbol x}).
\end{equation}
The standard errors of the MLE's based on the EM-algorithm are the square root of the diagonal elements of the $J(\hat{\boldsymbol\Theta};{\boldsymbol x})$. The computation of these matrices are too long and tedious. Therefore, we did not present the details. Reader can see \cite{ma-ja-12} how to calculate these values.

\section{A real example}
\label{sec.ex}

In this section, we analyze the real data set given by
\cite{mu-xi-ji-04}
to demonstrate the performance of EEWPS class of distributions in practice.
This data set consists of the failure times of 20 mechanical components, and is also studied by \cite{si-bo-di-co-13}:

\begin{center}
0.067, 0.068, 0.076, 0.081, 0.084, 0.085, 0.085, 0.086, 0.089, 0.098

0.098, 0.114, 0.114, 0.115, 0.121, 0.125, 0.131, 0.149, 0.160, 0.485
\end{center}

Since the EEWPS distribution can be used for  modeling of  failure times,
we consider  this distribution for fitting these data. But, this  distribution is a large class of distributions. Here, we consider five sub-models of EEWPS distribution. Some of them are suggested in literature.

\noindent i. The exponentiated Weibull geometric (EWG) distribution, i.e.  the EEWPS distribution with $H(x,\boldsymbol \Theta)=x^\gamma$  and $C(\lambda)=\lambda(1-\lambda)^{-1}$.

\noindent ii. The complementary Weibull geometric (CWG) distribution, i.e.  the EEWPS distribution with $H(x,\boldsymbol \Theta)=x^\gamma$, $C(\lambda)=\lambda(1-\lambda)^{-1}$ and $\beta=1$. This distribution  is  considered by \cite{co-si-14}.

\noindent iii. The generalized exponential geometric (GEG) distribution, i.e.  the EEWPS distribution with $H(x,\boldsymbol \Theta)=x$ and $C(\lambda)=\lambda(1-\lambda)^{-1}$. This distribution  is  considered by \cite{ma-ja-12}.

\noindent iv. The exponentiated Chen logarithmic (ECL) distribution, i.e.  the EEWPS distribution with $H(x,\boldsymbol \Theta)=\exp( x^\gamma)$ and $C(\lambda)=-\log(1-\lambda)$.

\noindent iv. The complementary Chen logarithmic (CCL) distribution, i.e.  the EEWPS distribution with $H(x,\boldsymbol \Theta)=\exp(x^\gamma)$, $C(\lambda)=-\log(1-\lambda)$ and $\beta=1$. This distribution  is  considered by \cite{co-si-14}.

The MLE's of the parameters  for the distributions are obtained by the EM algorithm given in Section \ref{sec.est}. Also, the standard errors of MLE's  are computed and given in paracenteses.
To test the goodness-of-fit of the distributions, we calculated the maximized log-likelihood ($\log (L)$), the Kolmogorov-Smirnov (K-S) statistic with its respective p-value, the AIC (Akaike Information Criterion), AICC (AIC with correction), BIC (Bayesian Information Criterion), CM (Cramer-von Mises statistic) and AD (Anderson-Darling statistic)  for the five submodels of distribution.
 The R software \citep{rdev-14}  is used for the computations.

 The results are given in Table \ref{table.EX1}, and from K-S, it can be concluded that all five models are appropriate for this data set.  But,  the EWG and ECL distributions are better than other distributions. In fact, we have a better fit when there is the parameter $\beta$ (exponentiated parameter) in model.
 The plots of the densities (together with the data histogram)  and cdf's given in Figure \ref{plot.EX1} confirm this conclusion.

\begin{table}[ht]
\begin{center}
\caption{Parameter estimates (standard errors), K-S statistic,
\textit{p}-value, AIC, AICC, BIC, CM and AD for the data set.}\label{table.EX1}
{\small
\begin{tabular}{|c|c|c|c|c|c|} \hline
Distribution & EWG & CWG & GEG & ECL & CCL \\ \hline
$\hat{\alpha }$ & 28.665 & 25.972 & 27.752 & 17.111 & 22.019 \\
(s.e.) & (4.617) & (11.093) & (6.841) & (1.572) & (10.177) \\ \hline
$\hat{\gamma }$ & 0.199 & 1.642 & --- & 0.136 & 1.586 \\
(s.e.) & (0.052) & (0.407) & --- & (0.026) & (0.231) \\ \hline
$\hat{\lambda }$ & 0.136 & 0.012 & 0.001 & 0. 146 & 0.261 \\
(s.e.) & ~(0.918) & (1.122) & (0.658) & ~(0.032) & (0.332) \\ \hline
$\hat{\beta }$ & $5.5\ e^7$ & --- & 13.825 & $7.4\ e^7$ & --- \\
(s.e.) & ($1.6\ e^8$) & --- & (8.471) & $2.3\ e^7$ & --- \\ \hline
${\log  \left(L\right)\ }$ & 37.978 & 26.422 & 32.976 & 37.794 & 25.759 \\
K-S & 0.124 & 0.264 & 0.160 & 0.121 & 0.262 \\
p-value & 0.917 & 0.122 & 0.683 & 0.931 & 0.127 \\
AIC & -67.957 & -46.845 & -59.952 & -67.588 & -45.518 \\
AICC & -65.29 & -45.345 & -58.452 & -64.922 & -44.018 \\
BIC & -63.974 & -43.858 & -56.965 & -63.606 & -42.531 \\
CM & 0.048 & 0.436 & 0.153 & 0.051 & 0.463 \\
AD & 0.402 & 2.537 & 1.136 & 0.423 & 2.663 \\ \hline
\end{tabular}

}
\end{center}
\end{table}

\begin{figure}
\centering
\includegraphics[width=7.75cm,height=7cm]{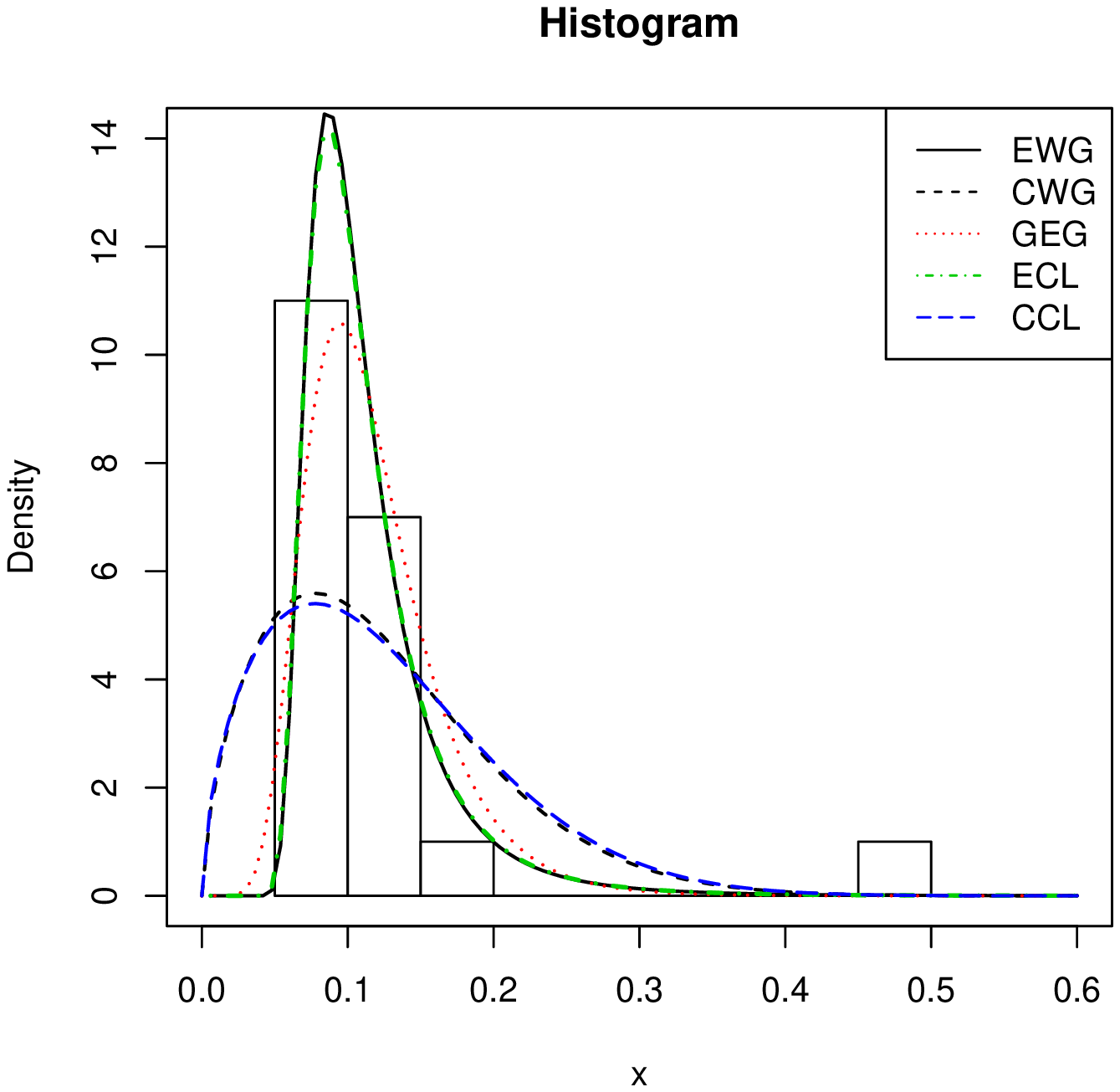}
\includegraphics[width=7.75cm,height=7cm]{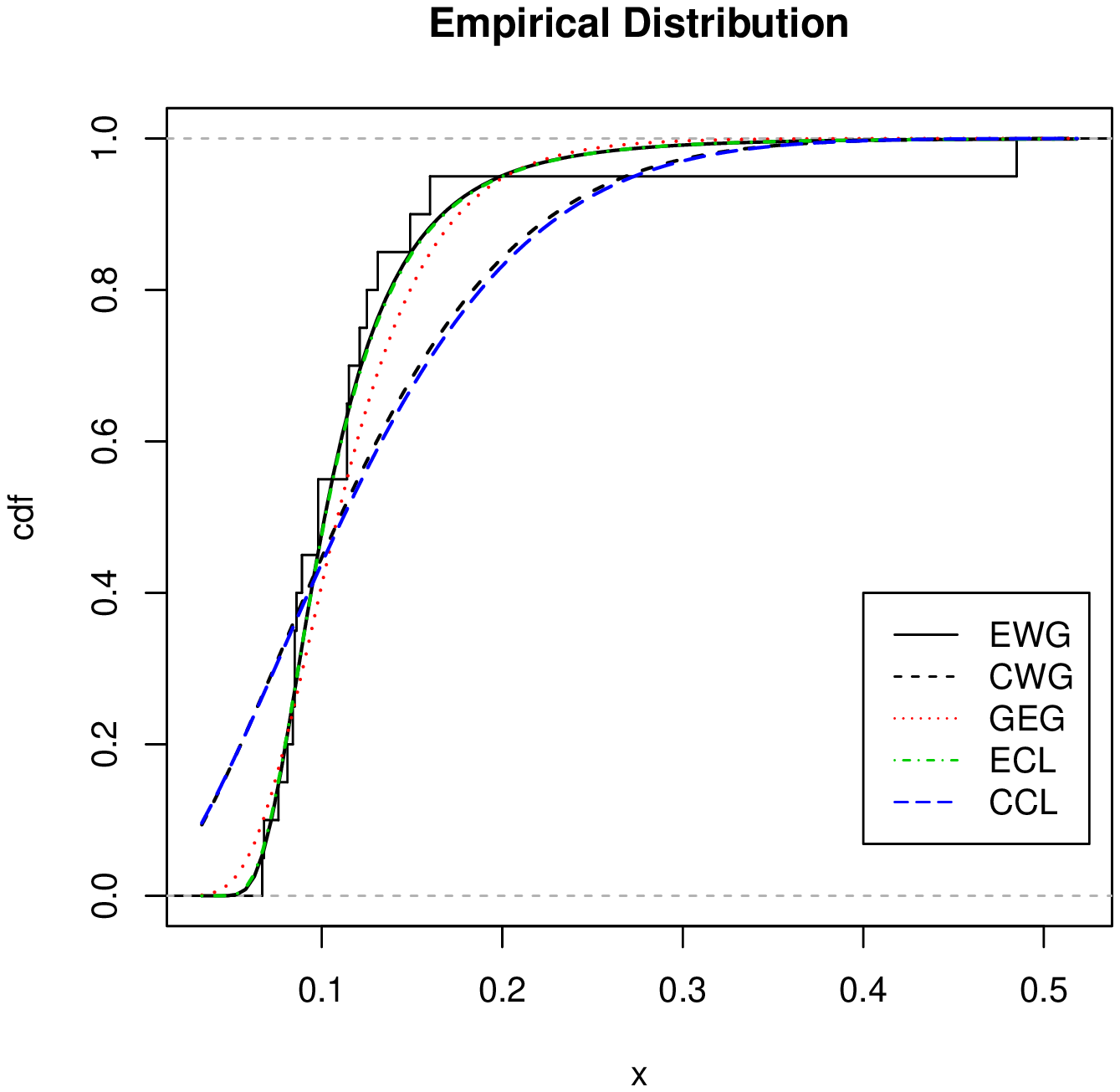}
\vspace{-1cm}
\caption{The histogram of the data set with the estimated pdf's (left), the empirical cdf of the data set, and estimated cdf's (right) for  fitted  of five submodels.}\label{plot.EX1}
\end{figure}

\section*{Acknowledgments}
The authors are thankful to the referees for helpful comments and suggestions.

\bibliographystyle{apa}

\end{document}